\documentclass[conference]{IEEEtran}
\usepackage{graphicx,subfigure}
\usepackage{enumerate}
\usepackage{color}
\usepackage{pifont}
\usepackage{float}
\usepackage{epsfig}
\usepackage{epstopdf}
\usepackage{cite,amssymb,amsthm}
\usepackage{tablefootnote}
\usepackage{amsmath}
\usepackage{arydshln}
\newtheorem{prop}{Proposition}

\newcommand{\mup}{\mu \text{PMU}}
\newcommand{\bs}{\boldsymbol} 

\newcommand{\mb}{\mathbf}

\begin{document}

\newenvironment{packed_enum}{
\begin{enumerate}
  \setlength{\itemsep}{1pt}
  \setlength{\parskip}{0pt}
  \setlength{\parsep}{0pt}
}{\end{enumerate}}

\newenvironment{packed_item}{
%\vspace{-3pt}
\begin{itemize}
  \setlength{\itemsep}{0pt}
  \setlength{\parskip}{0pt}
  \setlength{\parsep}{0pt}
}{\end{itemize}}

\title{Automated Anomaly Detection in Distribution Grids Using $\mup$ Measurements}
\author{Mahdi Jamei\IEEEauthorrefmark{1}, Anna Scaglione\IEEEauthorrefmark{1}, Ciaran Roberts\IEEEauthorrefmark{2},
Emma Stewart\IEEEauthorrefmark{2},\\
Sean Peisert\IEEEauthorrefmark{2},
Chuck McParland\IEEEauthorrefmark{2},
Alex McEachern\IEEEauthorrefmark{3},
\vspace{0.05cm}\\
\IEEEauthorrefmark{1}School of Electrical, Computer, and Energy Engineering, Arizona State University, Tempe, AZ, USA \\
\IEEEauthorrefmark{2}Lawrence Berkeley National Laboratory, Berkeley, CA, USA \\
\IEEEauthorrefmark{3}Power Standards Laboratory, Alameda, CA, USA}
\maketitle

\begin{abstract} 
The impact of Phasor Measurement Units (PMUs) for providing situational awareness to transmission system operators has been widely documented. 
Micro-PMUs ($\mup$s) are an emerging sensing technology that can provide similar benefits to Distribution System Operators (DSOs), enabling a level of visibility into the distribution grid that was previously unattainable. In order to support the deployment of these high resolution sensors, the automation of data analysis and prioritizing communication to the DSO becomes crucial.  
In this paper, we explore the use of $\mup$s to detect anomalies on the distribution grid. Our methodology is motivated by growing concern about failures and attacks to distribution automation equipment. The effectiveness of our approach is demonstrated through both real and simulated data.          
\end{abstract}
\IEEEpeerreviewmaketitle
\section{Introduction}
\label{sec.int}
The state vectors of the transmission grid are closely monitored and their physical behavior is well-understood \cite{abur2004power}. In contrast, Distribution System Operators (DSOs) have historically lacked detailed real-time actionable information about their system. This, however, is set to change. As the distribution grid shifts from a demand serving network towards an interactive grid, there is a growing interest in gaining situational awareness via advanced sensors such as Micro-Phasor Measurement Units ($\mup$s) \cite{scoping_study}. 

The deployment of the $\mup$s in isolation without additional data driven applications and analytics is insufficient. It is critical to equip DSOs with complimentary software tools that are capable of automatically mining these large data sets in search of useful, actionable information. There has been a lot of work focused on using PMU data at the transmission level to improve Wide-Area Monitoring, Protection and Control (WAMPC)\cite{phadke2008wide,terzija2011wide}. The distribution grid, however, is lagging in this respect. Due to inherent differences between operational behavior, such as imbalances and increased variability on the distribution and transmission grid, the algorithms derived for WAMPC at the transmission level are generally not directly applicable at the distribution level. Our work is aimed at addressing this issue.
%Therefore, in tandem with the growing number of sensors installed in the grid, researchers at universities and industries should develop algorithms that can make the most benefit out of this wealthy source of information, while considering the quite different operation environment that we have in the distribution grid.  
We focus on an important application of $\mup$ data in the distribution system: anomaly detection, i.e., behavior that differs significantly from normal operation of the grid during (quasi) steady-state. An anomaly can take a number of forms, including faults, misoperations of devices or switching transients, among others, and its root cause can be either a natural occurrence, error or attack. The risk of cyber-physical attacks via an IP network has recently gained significant interest due to the increase in automation of our power gird via two-way communication. This communication is typically carried out on breachable networks that can be manipulated by attackers\cite{slay2008lessons}. Even if an anomaly naturally occurs, it is important to notify the DSO to ensure proper remedial action is taken. 
\subsection{Related Work}
The majority of published work in anomaly detection using sensor data, primarily SCADA and PMU data, has focused on the transmission grid. The proposed methods are typically data-driven approaches, whereby the measurements are inspected for abnormality irrespective of the underlying physical model. One such example, the {\it common path} data mining approach implemented on PMU data and audit logs at a central server, is proposed in \cite{pan2015developing} to classify between a disturbance, an attack via IP computer networks and normal operation. Chen et al., \cite{chen2013dimensionality} derive a linear basis expansion for the PMU data to reduce the dimensionality of the measurements. Through this linear basis expansion, it is shown how an anomaly, which changes the grid operating point, can be spotted by comparing the error of the projected data onto the subspace spanned by the basis and the actual values. Valenzuela et al., \cite{valenzuela2013real} used Principal Component Analysis (PCA) to classify the power flow results into regular and irregular subspaces. Through analyzing the data residing in the irregular subspace, their method determines whether the irregularity is caused by a network attack or not. Jamei et al., \cite{jamei2016micro} propose an intrusion detection architecture that leverages $\mup$ data and SCADA communication over IP networks to detect potentially damaging activities in the grid.
These aforementioned algorithms are all part of the suite of machine learning techniques that the security monitoring architecture will rely on.  
\subsection{Our Contribution}
$\mup$s, due to their high sampling frequency, are a much richer data source in comparison to traditional Distribution Supervisory Control and Data Acquisition  (DSCADA). 
In this paper, we highlight the capability of $\mup$s to detect transient events by proposing a set of rules for anomaly detection. 
%These algorithms could be incorporated in an anomaly detection engine, to flag the abnormalities in the grid in an automated fashion, sifting for interesting patterns through the data even before they are sent to the DO historian. Their use can help DOs take  preventive/isolating/restorative actions, before the event spreads out. Our rules are based on the safety engineering practices that are common in the power quality domain as well as what is imposed by the governing equations of the grid i.e. Kirchhoff's and Ohm's law. 
The main advantages of our new approach are:\\
1) The underlying physical model of the data forms the basis in deriving the detection method; providing an interpretation of the event that is lacking in a model free approach.\\
2) The distribution grid is modeled allowing unbalanced loading and non-transposed lines. The rules are formulated in such a way that allow for distribution grids with neutral wires, and single-phase or two-phase laterals. \\
3) Quasi steady-state, rather than steady-state, is considered the norm for grid behavior.\\% Accordingly, its effects on the $\mup$ data and our rules are discussed in details.\\
4) Part of the proposed methodology only requires the phasor data stream of a single $\mup$ and is agnostic of the grid interconnection parameters, while the other part correlates the phasor streams across multiple $\mup$s using electrical properties of the grid. 

The detection method applicable to the measurements of a single $\mup$ is particularly attractive from a security perspective because, assuming that the algorithms are programmed onto the sensor itself, no network communication exchange is needed to obtain results. Therefore, the attacker will have to directly compromise the sensor to alter its response and erase evidence of a physical change.         

In addition to defining these algorithms, we explore their effectiveness in the field via an actual deployment of $\mup$s designed by our partners at Power Standards Lab. These devices output the three phase voltage and current phasors at specific locations on the distribution grid \cite{upmu_site} at a rate of 120 Hz. The proprietary filtering implemented within the device, which differs from the options for the filter given in the C.37.118 standard \cite{c37}, overcomes some of the technical obstacles limiting the deployment of conventional PMUs on the distribution grid. Technical obstacles in real world deployments include the presence of signal noise and smaller voltage angle differences. These devices are also inexpensive, which is a key feature for distribution sensors\cite{von2014micro}. 
We also investigate the sensitivity of our rules with respect to a set of attack scenarios on the $\mup$, and the grid connectivity data.

The remainder of the paper is organized as follows. Section~\ref{sec.model} introduces the $\mup$ data model. Section~\ref{sec:line_mod} presents the $\pi$ model of a distribution line, and the relationship between the voltage and current phasors in quasi steady-state conditions. Section~\ref{sec:rule} forms the body of our work that concerns itself with the formulation of the rules and tracking the anomalies. The effectiveness of the proposed rules, and their sensitivity to partially-compromised data, are tested using real and simulated data in Section~\ref{sec:res}. The conclusion follows in Section~\ref{sec:conc}.          
%%%%%%%%%%%%%%%%%%%%%%%%%%%%%%%%%%%%%%%%%%%%%%%%%%%%%%
\section{$\mup$ Data Modeling} 
\label{sec.model}
Assuming normal conditions, $\mup$s are designed to report 120 samples per sec. of the three phase voltage phasors, denoted as $\mb{v}[k]\in \mathbb{C}^{3\times 1}$, and the current phasors, $\mb{i}[k]\in \mathbb{C}^{3\times 1}$, at specific points on the distribution network. 
To help understand the effect of transients on $\mup$ measurements, and the difference between using phasor information compared to higher resolution time domain samples, we review the notion of a phasor, or complex envelope, in this section and tie its formal derivation to the actual implementation in practical $\mup$ sensors (including the C.37.118 standard \cite{c37}). 

The {\it phasor} is, in signal theory, often referred to as the {\it complex envelope}, or the {\it complex baseband equivalent representation}, of an arbitrary signal $s(t)$. In its textbook derivation, it is obtained in two steps. In the first step, the frequency content of the signal at negative frequencies is removed, which leads to a complex signal called an analytic signal, $s^+(t)$. The time domain mapping from $s(t)$ to $s^+(t)$ is as follows\footnote{The operator $\star$ stands for convolution.}:
\begin{equation}
s^+(t)=\frac 1 2 s(t)+\frac j{2 \pi t} \star s(t)
\label{eq.anal}
\end{equation}     
The second term in the sum is the {\it Hilbert transform} of the signal.  
The complex envelope, or phasor, can then be extracted by shifting down the analytic signal in the frequency domain and scaling it or, equivalently, demodulating and scaling the signal in the time-domain\footnote{The reason for the $\sqrt{2}$ scaling is that the signal and its envelope have the same power and energy.}:
\begin{align}\label{eq.env}
\tilde s(t)=\sqrt{2}s^+(t)e^{-j2\pi f_0 t}.
\end{align}  
From \eqref{eq.env} and \eqref{eq.anal} it is easy to see that:
\begin{equation}\label{eq.rec}
s(t)=\sqrt{2}\Re[\tilde s(t)e^{j 2 \pi f_0 t}].
\end{equation}

If power spectral density of the signal $s(t)$ is centered around $f_0$, then the complex representation is 
the {\it smoothest} signal that one can associate to $s(t)$. The mapping is one to one, and therefore, there is no loss of information. Actual $\mup$s do not perform the envelope this way, as explained in the following section. 

\subsection{Practical Implementations of $\mup$s}
The Hilbert transform requires implementing a non-causal filter with infinite impulse response, hence, it is purely theoretical. While there are other ways of approximating the Hilbert filter, the simplest implementation of the $\mup$ is based on the following observations. 
Replacing the operator that takes the real part of the modulated complex envelope in \eqref{eq.rec}, by the scaled summation of the modulated complex phasor and its conjugate and multiplying both sides of equation \eqref{eq.rec} by $\sqrt{2}e^{-j 2 \pi f_0 t}$, we have : 
% \begin{align}
% \sqrt{2}s(t)=\frac{\tilde s(t)e^{j 2 \pi f_0 t}+\tilde{s}^* (t)e^{-j 2 \pi f_0 t}}{2}
% \end{align}
%and equivalently we can write:
\begin{align}
\sqrt{2}s(t)e^{-j 2 \pi f_0 t}=\frac{\tilde s(t)+\tilde{s}^*(t)e^{-j 4 \pi f_0 t}}{2}
\label{eq:}
\end{align}
This representation is insightful since we can observe that a low-pass filter, $h(t)$,   
can extract the correct envelope $\tilde s(t)$ from the signal $\sqrt{2}s(t)e^{-j 2 \pi f_0 t}$ if and only if:
\begin{align}
&\tilde s(t) \star h(t)=\tilde s(t) \label{c.1}\\
&\tilde s^*(t)e^{-j 4 \pi f_0 t} \star h(t)=0 \label{c.2}
\end{align}
For the bandpass signal $s(t)$ centered around $f_0$, that has a limited support $[f_0-W_1,f_0+W_2]$, these two conditions are satisfied if:
\begin{align}\label{c.f0W}
W/2 \leq f_0,~~W=W_1+W_2.
\end{align} 
%and the corresponding to ideal filter has the following frequency response:
%\begin{align}
%\mathcal{F}(h(t))=H(f)\approx \begin{cases}
%1& -W_1<f<W_2\\
%0& \text{elsewhere} 
%\end{cases}
%\end{align} 
Accordingly, the mapping becomes:
\begin{align}\label{eq.pmuenv}
\tilde s(t)&=\sqrt{2}(s(t)e^{-j 2 \pi f_0 t})\star h(t)
\end{align} 
This means that the phasor can be extracted without obtaining the analytic signal if \eqref{c.1}-\eqref{c.2} hold, and the filter frequency response $H(f)$ is flat within the bandwidth of the signal and isolates its spectrum, emulating an ideal low pass filter. 
\subsection{Information from $\mup$s During a Transient} 
In normal conditions, the mains AC voltage and current are very close approximations of bandpass signals with a narrow support centered around the frequency $f_0=50/60$ Hz. In turn, \eqref{c.f0W} holds in the quasi steady-state scenario and for some types of transients. Because of the normally-small frequency support of the AC signals, the dynamic effects of the electrical wires are not apparent, and their effects can be approximated by a scaling and phase rotation equal to the amplitude and phase of their frequency response at the center frequency. This results in the well-known algebraic equations used in steady-state power systems analysis.
 
During a severe transient, however, the envelope that is obtained through \eqref{eq.env} is not the signal envelope, not just because frequency content is effectively cut by $h(t)$, but also because of the component in \eqref{c.2} that is not zero and the spilling of its tail into the band selected by $h(t)$ distorts the content. In either case, however distorted, the signal that emerges out of the filter $h(t)$ is band-limited, and can therefore be sampled without aliasing at a rate of $2f_0$ Hz. 

Our proposed methodology is to monitor whether the instantaneous $\mup$ measurements of voltage and current phasors belong to a complex hyperplane compatible with the algebraic steady-state version of Ohm's law for a three-phase unbalanced system. 
In the next section, we introduce the general Multi-Input Multi-Output (MIMO) representation of a distribution line and its quasi steady-state representation. These equations form the cornerstone of our anomaly detection rules.
%%%%%%%%%%%%%%%%%%%%%%%%%%%%%%%%%%%%%%%%%%%%%%%%%%%%%%%%  
\section{Distribution Line MIMO Modeling}
\label{sec:line_mod}
 
The $\pi$ model of a distribution line that connects bus $i$ to $j$ is shown in Fig.~\ref{fig:dis_line} where $\bs{Y}_{ij}(f)$ denotes the three phase series admittance of the line $(i,j)$ and $\bs{Y}^{sh}_{ij}(f)$ is the three phase shunt admittance matrix of that line.      
\begin{figure}
	\begin{center}
	\includegraphics[width=0.5\textwidth]{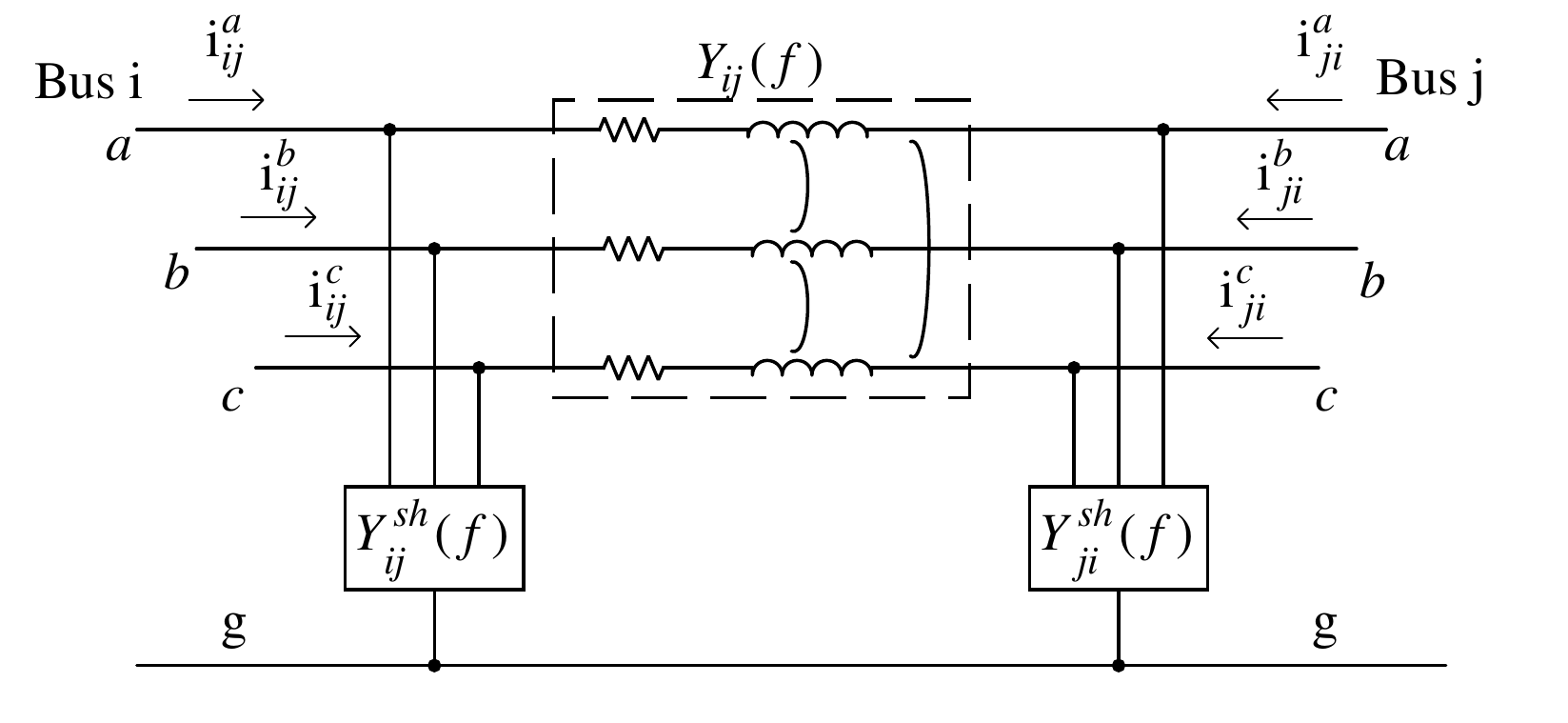}
	\caption{\small $\pi$ Model of the Distribution Line}
	\vspace{-0.2cm}
	\label{fig:dis_line}
	\end{center}
\end{figure}
The modeling of the self and mutual impedance, rather than using the positive sequence representation, is essential for accurately modeling the distribution grid since it does not impose the assumption of a balanced system nor transposition of the lines. In addition, this representation enables us to include three phase, two phase, and single phase lines with/without neutrals by using their $3 \times 3$ phase frame matrix model that is obtained from the modified Carson's equations and Kron reduction \cite{kersting2012distribution}.  

It is well known in linear systems theory that the relationship between the voltage and the current signals in a passive Linear-Time Invariant (LTI) circuit with a known admittance matrix can be represented  as the multiplication in the frequency domain and as a convolution in the time domain, which also takes MIMO form due to the three-phase modeling. The relationship also holds between the complex envelopes of the signals. 

Denoting the baseband model of the line admittance matrices frequency response around the center frequency as $\mathbf{y}_{ij}(f)=\bs{Y}_{ij}^{sh}(f+f_0)H(f)$, $\mathbf{y}^{sh}_{ij}(f)=\bs{Y}^{sh}_{ij}(f+f_0)H(f)$, we have:
\begin{align}
\label{eq:phasor_LTI1f}
{\bs{i}}_{ij}(f)&=(\mathbf{y}^{sh}_{ij}(f)+\mathbf{y}_{ij}(f))\bs{v}_i(f)-\mathbf{y}_{ij}(f)\bs{v}_j(f)\\
{\mathbf{i}}_{ij}(t)&=(\bs{y}^{sh}_{ij}(t)+\bs{y}_{ij}(t)) \ast  \mathbf{v}_i(t)-\bs{y}_{ij}(t) \ast \mathbf{v}_j(t)
\label{eq:phasor_LTI1t}
\end{align}
where $\bs{y}^{sh}_{ij}(t)$ and $\bs{y}_{ij}(t)$ are the time-domain equivalence of the baseband shunt and admittance matrices respectively.            
We will use the following notation for brevity in the remainder:
\begin{align*}
\begin{split}
\overline{\bs{Y}}_{ij}(f)&\triangleq
\bs{Y}_{ij}^{sh}(f)+\bs{Y}_{ij}(f)\\
\overline{\mathbf{y}}_{ij}(f)&\triangleq
\mathbf{y}^{sh}_{ij}(f)+\mathbf{y}_{ij}(f)\\
 \overline{\bs y}_{ij}(t)&\triangleq
 \bs{y}^{sh}_{ij}(t)+\bs{y}_{ij}(t)
 \end{split}
 \end{align*}
Using this defined notation, \eqref{eq:phasor_LTI1f} and \eqref{eq:phasor_LTI1t} can be re-written as:
 \begin{align}
\bs{i}_{ij}(f)&=\overline{\mathbf y}_{ij}(f) \bs{v}_i(f)-\mathbf{y}_{ij}(f) \bs{v}_j(f)\\
 \mathbf{i}_{ij}(t)&= \overline{\bs y}_{ij}(t) \ast \mathbf{v}_i(t)-\bs{y}_{ij}(t) \ast \mathbf{v}_j(t)
\label{eq:phasor_LTI}
\end{align}    
For the discrete time samples of the output of the $\mup$, the counterpart of \eqref{eq:phasor_LTI} is:
\begin{align}
\begin{split}
 \mathbf{i}_{ij}[k]&=\sum_{n=0}^{N-1} \overline{\bs y}_{ij}[n] \mathbf{v}_i[k-n]-\sum_{n=0}^{N-1}\bs{y}_{ij}[n]\mathbf{v}_j[k-n]
\end{split}
\label{eq:basiceq}
\end{align}
where the assumption is that $\bs{y}^{sh}_{ij}[n]$ and $\bs{y}_{ij}[n]$ are the samples of $\bs{y}^{sh}_{ij}(t)$ and $\bs{y}_{ij}(t)$, respectively, and that they are causal and have a finite support of $N$.

%In reality, the steady-state never happens in power system so it is important to understand   the impact on the $\mup$ samples to be able to distinguish between what is normal from what is not. 
In the quasi-steady state condition, the fundamental frequency of the voltage and current signals is always changing, albeit slowly and over a very small range, because of load-generation imbalances, active power demand interactions, large generators inertia, and the automatic speed controllers of the generators \cite{phadke2008synchronized}. 
As a result, the off-nominal frequency affects the phase angle captured by $\mup$s. Fig.~\ref{fig:unwrap} shows the three phase unwrapped angle of the voltage phasor data captured by a $\mup$ installed at our partner utility grid, which clearly shows the grid is working at off-nominal frequency where the frequency drift is not even fixed over time though varies slowly.      
\begin{figure}[h]
\includegraphics[width=0.5\textwidth]{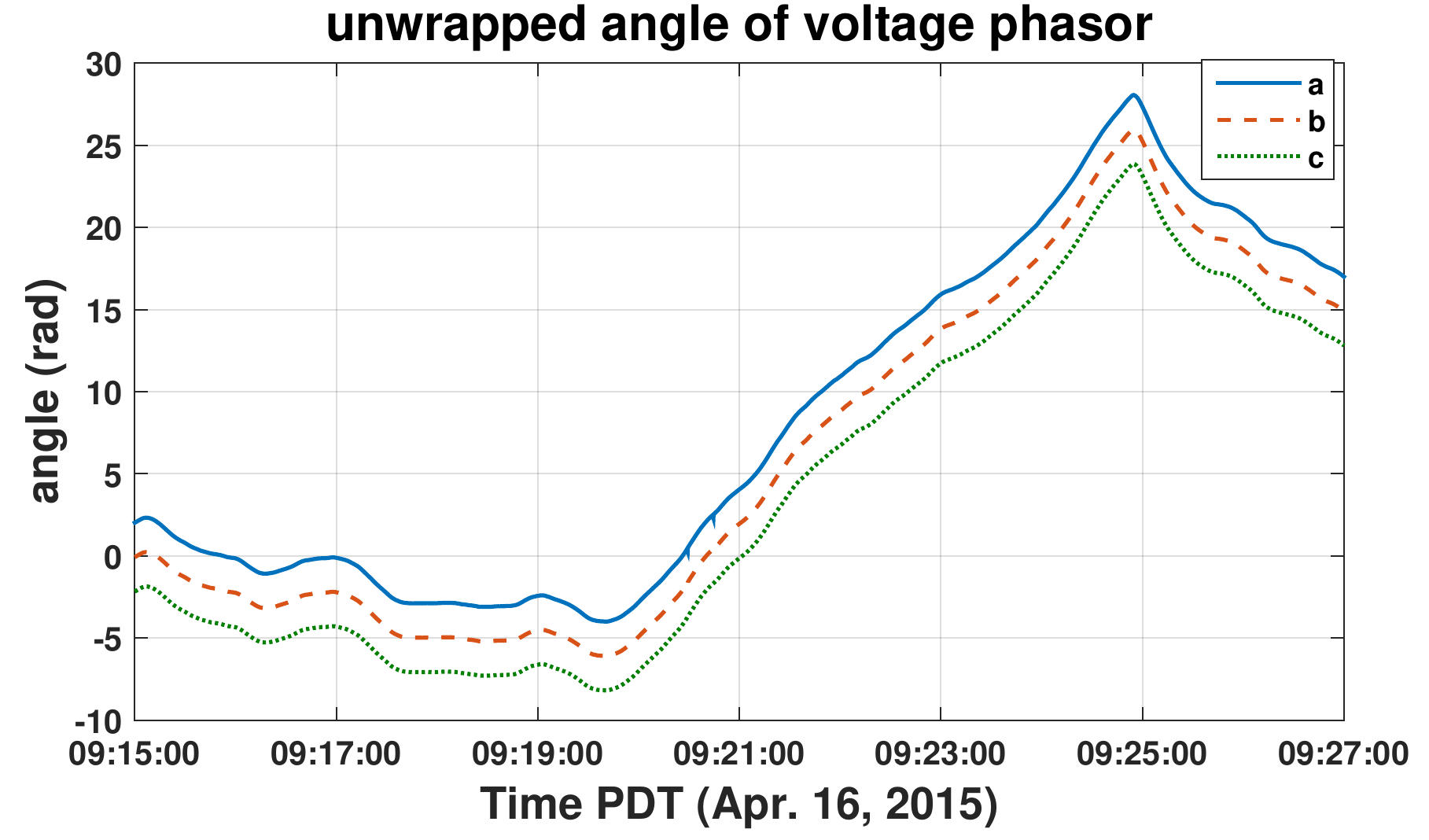}
\caption{Unwrapped Voltage Phasor Angle During Quasi-Steady State}
\label{fig:unwrap}
\end{figure}
Mathematically, it is insightful to decompose the phasor ${\mathbf{v}}_i[k]$ and $\mathbf{i}_{ij}[k]$ as follows:
\begin{align}
\mathbf{v}_i[k]=\hat{\mathbf{v}}_i[k] e^{j \beta_i[k]k}, 
~ \mathbf{i}_{ij}[k]=\hat{\mathbf{i}}_{ij}[k] e^{j \beta_i[k]k}
\label{eq:V_k_I_ij}
\end{align}
where  $\hat{\mathbf{v}}_i[k]$ is the voltage phasor that is captured at nominal frequency, and $\hat{\mathbf{i}}_{ij}[k]$ is the current phasor after removing the exponential term due to $\beta_i[k]$ that captures the time-varying drift in the frequency previously described. We can then write\eqref{eq:basiceq} in the following way:
\begin{align}
\begin{split}
\mathbf{i}_{ij}[k]&=
\sum_{n=0}^{N-1}\overline{\bs{y}}_{ij}[n]\hat{\mathbf{v}}_i[k-n] e^{j\beta_i[k][k-n] (k-n)}\\
&-\sum_{n=0}^{N-1}\bs{y}_{ij}[n]\hat{\mathbf{v}}_j[k-n]e^{j\beta_j[k][k-n] (k-n)}
\end{split}
\end{align}   
The variations of the $\hat{\mathbf{v}}_i[k]$ and $ \beta[k]$ can be approximately neglected over $N$ samples of the discrete time convolution in \eqref{eq:basiceq}, i.e., $\hat{\mathbf{v}}_i[k-n]\approx \hat{\mathbf{v}}_i[k]$ and $\beta_i[k-n]\approx \beta_i[k]$
for $n=0,\ldots,N-1$. Followed by this approximation, we can write:   
\begin{align}
\begin{split}
\mathbf{i}_{ij}[k]
&\approx 
\sum_{n=0}^{N-1}\overline{\bs{y}}_{ij}[n]\hat{\mathbf{v}}_i[k]e^{j\beta_i[k](k-n)} \\
&-\sum_{n=0}^{N-1}\bs{y}_{ij}[n]\hat{\mathbf{v}}_j[k]e^{j\beta_j[k] (k-n)} \\
&=\left(\sum_{n=0}^{N-1}\overline{\bs{y}}_{ij}[n]e^{-j\beta_i[k] n}\right)\mathbf{v}_i[k] \\
&-\left(\sum_{n=0}^{N-1}\bs{y}_{ij}[n]e^{-j\beta_j[k] n}\right) \mathbf{v}_j[k]\label{eq.i[k]1}
\end{split}
\end{align}         
Considering the following relationship:
\begin{equation}
\bs{Y}_{ij}(f+f_0)H(f)=rect(Tf)\sum_{n=0}^{N-1} T\bs y_{ij}[n] e^{-j2\pi nT f}
\end{equation}
where $T=1/120$ sec. is the sampling interval of the $\mup$, we can introduce the following two matrices:
\begin{align}
\begin{split}
\overline{\bs Y}_{ij}(f_0,k)&\triangleq \frac{1}{T} \overline{\bs{Y}}_{ij}\!\!\left(f_0+\frac{\beta_i[k]}{2\pi T}\right)H\!\left(\frac{\beta_i[k]}{2\pi T}\right),\\
{\bs{Y}}_{ij}(f_0,k)&\triangleq \frac{1}{T} \bs{Y}_{ij}\!\!\left(f_0+\frac{\beta_j[k]}{2\pi T}\right)H\!\left(\frac{\beta_j[k]}{2\pi T}\right),
\end{split}
\label{eq:modulated_Y}
\end{align}
and therefore we can write \eqref{eq.i[k]1} as follows:
\begin{align}
\mathbf{i}_{ij}[k] &=\overline{\bs Y}_{ij}(f_0,k)\mathbf{v}_i[k] -{\bs{Y}}_{ij}(f_0,k) \mathbf{v}_j[k]\label{eq.i[k]}
\end{align} 

Equation \eqref{eq.i[k]} is  Ohm's law in the phasor domain under the quasi-steady state and comprises part of our anomaly detection algorithm derived in the next section. The corresponding equation for the steady state can be simply obtained by setting $\beta[k]=0$ in \eqref{eq.i[k]}.  

From the analysis above, and specifically equation \eqref{eq.i[k]}, it is clear that the effect of the quasi-steady state in the phasor domain is the modulation of the admittances \eqref{eq:modulated_Y}. The effect is usually modest, as $\beta[k]$ is small. However, during a severe transient with frequency support in the order of 10 Hz, the relationship \eqref{eq.i[k]} with the matrices in \eqref{eq:modulated_Y} does not hold anymore, and an indication of transients in the phasor domain is, in reality, a manifestation of the full dynamic behavior in \eqref{eq:basiceq}.

\section{Anomaly Detection}
\label{sec:rule}
When the power grid is no longer in quasi-steady state conditions, the relationship between voltage and the current phasors manifests its full dynamic behavior. What we propose is to inspect the validity of the memoryless algebraic equations to flag the existence of transients in the grid. 
Equation \eqref{eq.i[k]} provides the basis for our rule.
\subsection{Single $\mup$ Metric}
Considering the line in Fig.~\ref{fig:dis_line}, we first assume that two $\mup$s are installed at both ends of a line, i.e., bus $i$ and $j$, which means $\mb{i}_{ij}[k]$, $\mb{i}_{ji}[k]$, $\mb{v}_i[k]$ and $\mb{v}_j[k]$ are all available and they can exchange their information locally. This case helps explaining the method that follows, which makes use of data from a single $\mup$.

We can cast the two equations that hold 
between the voltage and current at two ends of the three-phase line as follows:
\begin{align}
\begin{split}
\underbrace{
\begin{pmatrix}
\mb{i}_{ij}[k]\\
\mb{i}_{ji}[k]
\end{pmatrix}}&=
\underbrace{\begin{pmatrix}
\overline{\bs Y}_{ij}(f_0,k)&-{\bs Y}_{ij}(f_0,k)\\
-{\bs Y}_{ij}(f_0,k)&\overline{\bs Y}_{ij}(f_0,k)
\end{pmatrix}}
\underbrace{
\begin{pmatrix}
\mb{v}_{i}[k]\\
\mb{v}_{j}[k]
\end{pmatrix}}\\
\mb{I}_{ij}[k]&~~~~~~~~~~~~~~~~\mb{H}_{ij}(f_0,k)~~~~~~~~~~~~~~\mb{V}_{ij}[k]
\end{split}
\label{eq.2mups}
\end{align}
%By analyzing the algebraic properties of the following two matrices obtained by computing current and voltage phasors cross and autocorrelation samples matrices (with $M\geq 6$), the validity of the quasi-steady state can be verified:
Let us define, with $M\geq 6$, the following sample correlation matrices:
 \begin{align}
 \bs R_{IV}[k]&=\frac{1}{M-1}\sum_{m=0}^{M-1} \mb{I}_{ij}[k-m]\mb{V}^H_{ij}[k-m],\\
  \bs R_{VV}[k]&=\frac{1}{M-1} \sum_{m=0}^{M-1} \mb{V}_{ij}[k-m]\mb{V}^H_{ij}[k-m].
 \end{align}
Assuming that variations of $\mb{H}_{ij}(f_0,k)$ is negligible over a window of $M$ samples, equation \eqref{eq.2mups} implies that the following homogeneous equation holds in quasi-steady:
 \begin{align}\label{eq.homog} 
\left( \begin{array}{c:c}
\mathcal{I}_6 &-\mb{H}_{ij}(f_0,k)
\end{array}\right)
\underbrace{
    \begin{pmatrix}
  \bs R_{ IV}[k]\\
    \bs R_{ VV}[k]
 \end{pmatrix}}_{\bs R_k}
 =\mb 0
  \end{align}   
%which, in turn, means that $\bs R_k$ must lose rank. It should be noted that a rank deficient correlation matrix $\bs R_k$ is normal for a balanced three-phase circuit
%because $ \bs R_{IV}[k]$ and $\bs R_{VV}[k]$ lose rank themselves. In quasi-steady state, however, the structure in \eqref{eq.homog} along with the temporal data correlation (that leads to have rank-deficient $ \bs R_{IV}[k]$ and $\bs R_{VV}[k]$) imply that the whole $\bs R_k$ has rank 1 or close to it, that is:
 \begin{prop}
 \label{prop.Rk2}
Correlation matrix $\bs R_k$ is approximately rank-1 during the quasi-steady state. 
 \end{prop}
\begin{proof}
During the quasi-steady state along a distribution line, the following assumptions hold with a very good approximation for $n=0,1,...,M-1$:
\begin{align}
\begin{split}
\hat{\mathbf{v}}_i[k-n]&\approx \hat{\mathbf{v}}_i[k],~~~\hat{\mathbf{v}}_j[k-n] \approx \hat{\mathbf{v}}_j[k]\\
\beta_i[k-n]&\approx \beta_i[k],~~~\beta_j[k-n]\approx \beta_j[k],\\
\beta_i[k] & \approx \beta_j[k]
\end{split}
\end{align}  
Therefore, we can write:
\begin{align}
\begin{split}
\bs R_{VV}[k]&=\frac{1}{M-1}(\mathbf{V}_{ij}[k] \otimes \mb{E}[k])(\mathbf{V}^H_{ij}[k]\otimes \mb{E}^H[k])\\
&=\frac{1}{M-1}(\mathbf{V}_{ij}[k]\mathbf{V}^H_{ij}[k])\otimes(\mb{E}[k]\mb{E}^H[k])
\end{split}
\end{align}
where $\mb{E}[k]$ is defined as follows and represents the variations due to the off-nominal frequency:
\begin{align}
\mb{E}[k]=\mb{1}_{6\times 1}\otimes\begin{pmatrix}
e^{-j\beta_i[k](M-1)}&\ldots&e^{-j\beta_i[k]}&1\\
\end{pmatrix}
\end{align}
We can then write:
\begin{align}
\mb{E}[k]\mb{E}^H[k]=(\mb{1}_{6\times 1}\mb{1}_{1\times 6})\otimes(M)=M\mb{1}_{6\times 6}
\end{align}
and therefore:
\begin{align}
\bs R_{VV}[k]=\frac{M}{M-1}(\mathbf{V}_{ij}[k]\mathbf{V}^H_{ij}[k])\otimes(\mb{1}_{6 \times 6})
\end{align}
which accordingly means that:
\begin{align*}
rank(\bs R_{VV}[k])=rank(\mathbf{V}_{ij}[k]\mathbf{V}^H_{ij}[k])\times rank(\mb{1}_{6 \times 6})=1 
\end{align*}
Because:
\begin{align*}
rank(\bs R_k)=rank(\bs R_k^H\bs R_k)
\end{align*}
We analyze the rank of $\bs R_k^H\bs R_k$ here, where:
\begin{align}
\bs R_k^H\bs R_k=\bs R_{IV}^H[k]\bs R_{IV}[k]+\bs R_{VV}^H[k]\bs R_{VV}[k]
\label{eq.RhR2}
\end{align}
From the structure of \eqref{eq.homog} during the quasi-steady state, we have:
\begin{align}
\begin{split}
\bs R_{IV}[k]&=\mb{H}_{ij}(f_0,k)\bs R_{VV}[k]
\end{split}
\label{eq.Htilde}
\end{align}
Substituting \eqref{eq.Htilde} in \eqref{eq.RhR2}, we have:
\begin{align}
\bs R_k^H\bs R_k=\bs R^H_{VV}[k]\bs{\mathcal{G}}_{ij}(f_0,k)\bs R_{VV}[k]
\end{align}
where:
\begin{align*}
\bs{\mathcal{G}}_{ij}(f_0,k)=\mb{H}^H_{ij}(f_0,k)\mb{H}_{ij}(f_0,k)+\mathcal{I} 
\end{align*}
%that is a positive-semidefinite matrix, and therefore we can write:
%\begin{align}
%\bs{\mathcal{H}}_{ij}(f_0,k)=\mb{L}_{ij}^H(f_0,k)\mb{L}_{ij}(f_0,k)
%\end{align}
%where $\mb{L}_{ij}(f_0,k)$ is the square-root of $\bs{\mathcal{H}}_{ij}(f_0,k)$, and therefore:
%\begin{align}
%rank(\bs R^H_k\bs R_k)=rank(\mb{L}_{ij}(f_0,k) \bs R_{VV}[k])
%\end{align} 
Since the linear transformation of $\bs R_{VV}[k]$ does not increase its rank, and since we have already shown that $\bs R_{VV}[k]$ is of rank-1 during the quasi-steady state, we can conclude that:
\begin{align}
\begin{split}
&rank(\bs R_k)=rank(\bs R_k^H\bs R_k)=\\
&\leq rank(\bs R_{VV}[k])=1 \rightarrow \\
 &rank(\bs R_k)=1
\end{split}
\end{align}  
\end{proof}
Therefore, we can write $\bs R_k$ as follows:    
\begin{eqnarray}
\bs R_k\approx \sigma_1[k]\bs u_1[k]\bs \nu^H_1[k]\rightarrow\bs R_k\bs R^H_k\approx \sigma^2_1[k]\bs u_1[k]\bs u^H_1[k],
\end{eqnarray}
which means that all columns of $\bs R_k$ must lie in the same hyperplane. 
Overall, deviation from this behavior is an indicator that the line is experiencing a transient and/or that the three-phase measurements no longer lie over a single principal component. 
The detection can be automated by computing the following cost  
%that is the size of the null-space of $\bs R_k\bs R_k^H$
and tracking the fast changes in $x[k]$ \footnote{$||.||_{F}$ denotes the Frobenious norm.}:
\begin{equation}
\begin{split}
x[k]&=\min_{\bs u} {||(\mathcal{I}_{12}-\bs u\bs u^H)\bs R_k\bs R_k^H||_{F}} ~~\mbox{s.t.}~~||\bs u||=1.
\end{split}
\label{eq:cost_two_upmu}
\end{equation}
The solution $\bs u$  is the principal subspace of $\bs R_k$ and it is normally obtained by minimizing the 
orthogonal projection with respect to $\bs u$ that is expected to ideally go to zero in the stationary balanced case, and be close to zero for stationary unbalanced case. 

Assume now that only a single $\mup$ is available at one end of a line. Using some reasonable approximations, it is still possible to apply this rule using the data stream from a single $\mup$. 
For bus $i$ with a $\mup$, the two voltage vectors at the two ends of each incident line to that bus are such that:
\begin{align}
\mb{v}_j[k]= \underset{\tilde{\bs{\alpha}}[k]}{\underbrace{\text{diag}(\bs{\alpha}[k])}} \mb{v}_i[k],
\end{align}
where $\bs{\alpha}[k]$ is a complex vector that relates the voltage phasors at the two ends. If we define now:
 \begin{align}
 \bs R^{(ij)}_{\rm iv}[k]&=\frac{1}{M-1}\sum_{m=0}^{M-1} \mb{i}_{ij}[k-m]\mb{v}^H_{i}[k-m],\\
  \bs R^{(ji)}_{\rm vv}[k]&=\frac{1}{M-1}\sum_{m=0}^{M-1} \mb{v}_{j}[k-m]\mb{v}^H_{i}[k-m].
 \end{align}
Assuming that $\bs{\alpha}[k]$ remains constant over a window of $M$ samples during the quasi-steady state, we can write:
 \begin{align}
\bs R^{(ji)}_{\rm vv}[k]\approx \tilde{\bs{\alpha}}[k] \bs R^{(ii)}_{\rm vv}[k]
 \end{align}
Assuming that the variation of $\overline{\bs Y}_{ij}(f_0,k)$ is negligible over $M$ samples during the quasi-steady state, we can use \eqref{eq.i[k]} to write:
\begin{align}
\left( \begin{array}{c:c}
\mathcal{I}_3 &
-\overline{\bs Y}_{ij}(f_0,k)+{\bs Y}_{ij}(f_0,k)\tilde{\bs{\alpha}}[k] 
\end{array}\right)
\underbrace{
 \begin{pmatrix}
  \bs R^{(ij)}_{\rm iv}[k]\\
    \bs R^{(ii)}_{\rm vv}[k]
 \end{pmatrix}
 }_{\bs R^{(i)}_k}
\approx\bs 0
\label{eq.homog1}
\end{align}
\begin{prop}
\label{prop.Rk1}
Correlation matrix $\bs R^{(i)}_k$ is approximately rank-1 during the quasi-steady state.
\end{prop}
\begin{proof}
The proof is very similar to that of Proposition~\ref{prop.Rk2}, and follows by assuming that $\hat{\mathbf{v}}_i[k-n]\approx \hat{\mathbf{v}}_i[k]$, $\beta_i[k-n]\approx \beta_i[k]$, and $\bs{\alpha}[k-n] \approx \bs{\alpha}[k]$ for $n=0,1,...,M-1$, as well as using the structure of \eqref{eq.homog1}.  
\end{proof}
Proposition~\ref{prop.Rk1} suggests a similar criterion for a single $\mup$ to flag the exit from a quasi-steady state regime, and can be achieved by tracking the fast changes in $x[k]$ defined as follows for each individual incident line to that bus:
\begin{equation}
x[k]=\min_{\bs u}{||(\mathcal{I}_{6}-\bs u\bs u^H)\bs R^{(i)}_k(\bs R^{(i)}_k)^H||_{F}} ~~\mbox{s.t.}~~||\bs u||=1
\label{eq:cost_one_upmu}
\end{equation}
\subsection{Multiple $\mup$s Metric} 
\label{sec:multiple}
In this section, we correlate the phasor data across multiple $\mup$s to detect anomalies in the grid. The applied rule here extends the test of the quasi-steady state equations validity applying it to multiple  $\mup$ measurements scattered over the grid. The rule can be hosted in the Distribution Management System (DMS), where the data from all the $\mup$s could be available, or it can be decentralized over a set of anomaly detection engines, where each agent is responsible to check the anomalies on a dedicated part of the grid and sharing the edge information with the other agents. 
%However, what we explain in the following is not focused on the decentralization, and is expressed generally.         

We assume that knowledge of $\bs{Y}^{sh}_{ij}(f_0,0)$ and $\bs{Y}_{ij}(f_0,0)$ for each line in the perimeter monitored by that detector engine is available. We can take advantage of the fact that $\beta_k$ is small and consider their difference from \eqref{eq:modulated_Y} as a perturbation, which is equivalent to noise in the observation model. For brevity, when introducing the rule in this part, we will use $\bs{Y}^{sh}_{ij}$ and $\bs{Y}_{ij}$ to refer to
$\bs{Y}^{sh}_{ij}(f_0,0)$ and $\bs{Y}_{ij}(f_0,0)$.

A natural way to relate the measurements across multiple devices is through the grid interconnection.  
We represent the vector of three-phase current injection and bus voltage phasors in the whole grid by $\mathbf{I}[k]$ and $\mathbf{V}[k]$, respectively, each vector contains $3B$ elements where $B$ is the number of buses. We also define the vector $\mathbf{d}$ as follows:
\begin{align}
\mathbf{d}[k]=\begin{pmatrix}
\mathbf{I}[k] \\ \mathbf{V}[k]
\end{pmatrix}
\end{align}
The following set of algebraic equations are homogeneous during the steady-state and should be close to homogeneity during the quasi-steady state:
\begin{align}
\mathbf{H}\mathbf{d}[k]
=\mathbf{0},~~~ 
\mathbf{H}&=\left( \begin{array}{c:c}
\mathcal{I}_{3B}  &-\bs{Y}_{3(B \times B)}
\end{array} \right)
\label{eq:grid_homogen}
\end{align} 
where $\bs{Y}$ is the admittance matrix of the grid that connects the current injection to the bus voltages, and is constructed from the $3 \times 3$ line shunt and series admittance matrices introduced in Section~\ref{sec:line_mod}.

It should be noted that the elements of $\mathbf{d}$ in \eqref{eq:grid_homogen} are not all independent variables. The challenge is that, in general, we will have a very limited number of measurements of $\mathbf{d}$ from $\mup$s, due to the cost limitations of deploying these devices. 
Let $K$ denote the number of $\mup$s that are available. We assume that each $\mup$ device has enough channels to measure the voltage and all incident current measurements of the bus on which it is installed. Hence, having a $\mup$ at bus $i$ means that the following three phase voltage phasors and three phase current injection phasors for that bus are available:
\begin{align}
\begin{split}
[\mb{V}[k]]_i=\mb{v}_i[k],~~
[\mb{I}[k]]_i&=\sum_{j:i \sim j} \mb{i}_{ij}[k]
\end{split}
\end{align}
where $i \sim j$ denotes that bus $i$ and $j$ are connected. We can define a permutation matrix ${\bf T}$
that parses the vector $\mathbf d[k]$ into two sub-vectors corresponding to the non-available measurements, $\mathbf{d}_u[k]$, and the available measurements, $\mathbf{d}_a[k]$, that is:
\begin{equation}
{\mathbf T}=
\begin{pmatrix}
{\mathbf T}_u\\
{\mathbf T}_a
\end{pmatrix}
~\rightarrow~ 
{\mathbf T}\mathbf d=\begin{pmatrix}
{\mathbf d}_u\\
{\mathbf d}_a
\end{pmatrix},
~\mathbf{H}{\mathbf T}^T=
\left(\!\!
\begin{array}{c:c}
{\mathbf H}_u&{\mathbf H}_a
\end{array}\!\!
\right)
\end{equation}
where
\begin{align*}
 \mb{T}_u &\in \{0,1\}^{6(K' \times B)},~K'=B-K\\
 \mb{T}_a &\in \{0,1\}^{6(K \times B)}
\end{align*}
Since ${\mathbf T}^T{\mathbf T}={\cal I}$, we can rewrite \eqref{eq:grid_homogen} in the following form:
\begin{align}
\mathbf{H}_u\mathbf{d}_u[k]+\mathbf{H}_a\mathbf{d}_a[k]=\mathbf{0}.
\label{eq:split_grid_homogen_1}
\end{align}  
Even though the equation is not exactly homogeneous, primarily due to the frequency drift discussed in Section \ref{sec:line_mod}, it suggests that an estimate of $\mb{d}_u$ can be found through the following minimization:
\begin{align}
\hat{x}[k]=\underset{\mathbf{d}_u}{\min}~ {{||\mathbf{H}_u\mathbf{d}_u+\mathbf{H}_a\mathbf{d}_a[k]||_2^2}}.
\label{eq:opt}
\end{align}
This is a least-square problem, with the well known solution:
\begin{align}
\mathbf{d}^\text{opt}_u[k]=-\mathbf{H}_u^\dagger\mathbf{H}_a\mathbf{d}_a[k]
\end{align}
where $(.)^\dagger$ denotes the pseudo-inverse operator. Consequently:
\begin{align}\label{xopt1}
\hat{x}[k]=||(\mathcal{I}-\mathbf{H}_u\mathbf{H}_u^\dagger)\mathbf{H}_a\mathbf{d}_a[k]||_2^2,
\end{align}
In fact, $\hat{x}[k]$ can be interpreted as follows. If we pre-multiply both sides of \eqref{eq:split_grid_homogen_1} by the orthogonal projector onto the left null-space of $\mb{H}_u$ i.e., $(\mathcal{I}-\mathbf{H}_u\mathbf{H}_u^\dagger)$, by definition, for any $\mathbf{d}_u[k]$, we have:
\begin{equation}
(\mathcal{I}-\mathbf{H}_u\mathbf{H}_u^\dagger)\mathbf{H}_u\mathbf{d}_u[k]=\mathbf{0}
\end{equation}
and therefore, if \eqref{eq:split_grid_homogen_1} holds, it must also hold that:
\begin{equation}\label{eq.null}
(\mathcal{I}-\mathbf{H}_u\mathbf{H}_u^\dagger)\mathbf{H}_a\mathbf{d}_a[k]=\mathbf{0}
\end{equation}
which explains why the cost in \eqref{xopt1} should be small when the system is in quasi-steady state. What we propose is to track the fast changes of the normalized $\hat{x}[k]$ defined as follows to detect anomalies in the data.
\begin{align}\label{xopt}
x[k]=\frac{||(\mathcal{I}-\mathbf{H}_u\mathbf{H}_u^\dagger)\mathbf{H}_a\mathbf{d}_a[k]||_2^2}{||\mb{d}_a||_2^2}
\end{align} 
However, \eqref{eq.null} becomes trivial if $(\mathcal{I}-\mathbf{H}_u\mathbf{H}_u^\dagger)=0$. This is the case when $\mb{H}_u$ is a full rank square matrix or a fat matrix, i.e., $K < \frac{B}{2}$, which has full row rank. This is actually the most common case because the number of $\mup$s is going to be generally very small relative to the grid size. However, we can rely on the fact that the matrix $\mb{H}_u\mb{H}^H_u$ has a high condition number, due to the weak connectivity of the radial or weakly meshed networks, and relative homogeneity of the line parameters. Considering the singular value decomposition of the matrix $\mb{H}_u$:
\begin{align}
\mb{H}_u=\mb{U}_u \mb{S}_u \mb{V}_u^H
\end{align}
We define $\mb{u}_{u,2}$ to denote the last column of $\mb{U}_u$, representing the left singular vector that corresponds to the smallest singular value of the matrix $\mb{H}_u$. If \eqref{eq:split_grid_homogen_1} holds, we can expect that the $x[k]$ defined as follows should be small with smooth variations during the quasi-steady state:
\begin{align}
x[k]=\frac{||\mb{u}^H_{u,2}\mb{H}_a\mb{d}_a||_2^2}{||\mb{d}_a||_2^2}
\label{xopt_underdet}
\end{align} 

%In quasi steady-state \eqref{eq:xopt_underdet} is bounded by the minimum singular value-squared of the matrix $\mb{H}_u$, which of    
Therefore, we propose $x[k]$ as the quantity to track its sudden changes in order to detect that a transient is present and apparent through the measurements $\mathbf{d}_a[k]$ or equivalently to spot when \eqref{eq:split_grid_homogen_1} no longer holds. 
%%%%%%%%%%%%%%%%%%%%%%%%%%%%%%%%%%%%%%%%%%%%%%%%%%%%%%%%
\subsection{Fast Change Detection Method}
\label{sec:change_det}
The quantities defined in the last section for our metrics are tracked for fast changes, since their sudden variations are signatures of an anomaly. The quantities consist of the optimum cost functions defined in \eqref{eq:cost_two_upmu} and \eqref{eq:cost_one_upmu} for the steady-state criterion using the data from two adjacent $\mup$s and a single $\mup$, respectively, and finally, the value defined in \eqref{xopt_underdet} using the data from multiple $\mup$s. 

Using real and simulated data, we have confirmed that changes in the mean value during the quasi steady-state regime are extremely smooth while deviations from this mean value are minimal. This observation motivated us to consider changes in their mean value as the common statistical trade-mark of anomalies in all of these quantities. 

To achieve fast detection of such changes, we propose to apply the sequential two-sided Cumulative Sum (CUSUM) algorithm \cite{page1954continuous,basseville1993detection,lai1995sequential}. To cast our problem in the CUSUM frame, we approximate the samples $x[k]$ for all the aforementioned quantities as outcomes of a Gaussian non-zero mean process ${X}[k]$:  
\begin{align}
{X}[k]={\mu}_x[k]+{\omega}_x[k]
\end{align}
where ${\mu}_x[k]$ is the mean of the process, and ${\omega}_x[k] \sim \mathcal{N}({0},\sigma^2_x[k])$ is zero-mean, additive Gaussian noise in the measurements with variance $\sigma^2_x[k]$. We assume that ${\omega}_x[k]$ are independent random variables for each $k$. Our intention is to detect sudden changes in the mean of the process, ${\mu}_x[k]$. Although there is temporal correlation among the observations, our objective, i.e., fast change detection in mean, justifies the relaxation that the random process has independent observation samples.   

The algorithm decides between two possible hypotheses at time $k$:
 $\mathcal{H}_0$: no change is detected in the mean,
 $\mathcal{H}_1$: change is detected in the mean.
%and if the change is detected (hypothesis $\mathcal{H}_1$), the question is when this change has happened. 
The decision in the CUSUM algorithm is based on two \textit{``instantaneous log-likelihood ratios''}, corresponding to upward and downward change of the mean, defined as follows:
\begin{subequations}
\begin{eqnarray}
\label{eq:si}
\lambda_X^u[k]=+\frac{|\widehat{\delta}_x|}{\widehat{\sigma}_x^2[k]}\left(x[k]-\widehat{\mu}_{0,x}[k]-\frac{|\widehat{\delta}_x|}{2}\right)\\
\lambda_X^d[k]=-\frac{|\widehat{\delta}_x|}{\widehat{\sigma}_x^2[k]}\left(x[k]-\widehat{\mu}_{0,x}[k]+\frac{|\widehat{\delta}_x|}{2}\right)
\label{eq:sd}
\end{eqnarray}
\end{subequations}
where the superscripts $u$ and  $d$ represent the variables corresponding to the ``upward'' and ``downward'' change detection respectively. $\widehat{\delta}_x$ is the mean change estimate, and is initialized based on \textit{``a priori''} knowledge. $\widehat{\sigma}_x^2[k]$ is the random process variance estimate that is assumed to remain constant during the change and $\widehat{\mu}_{0,x}[k]$ is the mean estimate. The mean estimate is obtained adaptively from normal ensembles using the exponential window as follows:
\begin{align}
\label{eq:mu}
\widehat{\mu}_{0,x}[k]=&\rho~\widehat{\mu}_{0,x}[k-1]+(1-\rho)x[k]
\end{align}
%\widehat{\sigma}_x^2[k]=&\gamma~\widehat{\sigma}_x^2[k-1]+ (1-\gamma)(x[k]-\widehat{\mu}_{0,x}[k])^2
where $0 \leq \rho \leq 1$ determines the dependency of the mean estimator on the past samples compared to the current sample respectively. Accordingly, two cumulative sums, ${M}_X^{u}[k]$ and ${M}_X^{d}[k]$, and two decision functions, ${G}_X^{u}[k]$ and ${G}_X^{d}[k]$, are derived as follows.
\begin{subequations}
\begin{align}
{M}_X^{u}[k]=&{M}_X^{u}[k-1]+{\lambda}_X^{u}[k] 
\label{eq:Si}
\\
{M}_X^{d}[k]=&{M}_X^{d}[k-1]+{\lambda}_X^{d}[k]
\label{eq:Sd}
\\
{G}_X^{u}[k]=&\max ({G}_X^{u}[k-1]+{\lambda}_X^{u}[k],{0}) 
\label{eq:Gi}
\\
{G}_X^{d}[k]=&\max ({G}_X^{d}[k-1]+{\lambda}_X^{d}[k],{0})
\label{eq:Gd} 
\end{align}
\end{subequations} 
The decision functions are then compared to a user-defined threshold, ${\alpha}_x$. Then, hypothesis $\mathcal{H}_1$ for the upward or downward change in the mean is chosen if either ${G}_X^u[k] > \alpha_x$ or $ {G}_X^d[k] > \alpha_x$, respectively. Depending on the change direction, the estimate of the anomaly start time is:
\begin{align}
\widehat{k}_{c,x}=\underset{{k}_{0,x} \leq k_{c,x} \leq k-1}{\text{argmin }}{M}_X^{u/d}[k_{c,x}]
\end{align}
where ${k}_{0,x}$ is the last detected change time index. 

During an event, we expect to see multiple change points. Detection of multiple changes is done by resetting the decision functions and cumulative sums to zero after the change is detected, and continuing the dynamic rule for upcoming samples. The fast change anomaly is completed if no new changes are detected for a defined window of time.
%%%%%%%%%%%%%%%%%%%%%%%%%%%%%%%%%%%%%%%%%%%%%%%%%%%%%%%%%
\section{Numerical Results}
\label{sec:res}
\subsection{Field Recorded Data}
Field recorded data is obtained from $\mup$s that are installed in our partner utility medium voltage (12.47 kV) grid. The window of data that we use to validate our detection rules contains a voltage sag event recorded on the network. The captured three phase voltage phasor magnitude by one of the $\mup$s during this event is shown in Fig.~\ref{fig:vol_sag}.
\begin{figure}[h]
\includegraphics[width=0.5\textwidth]{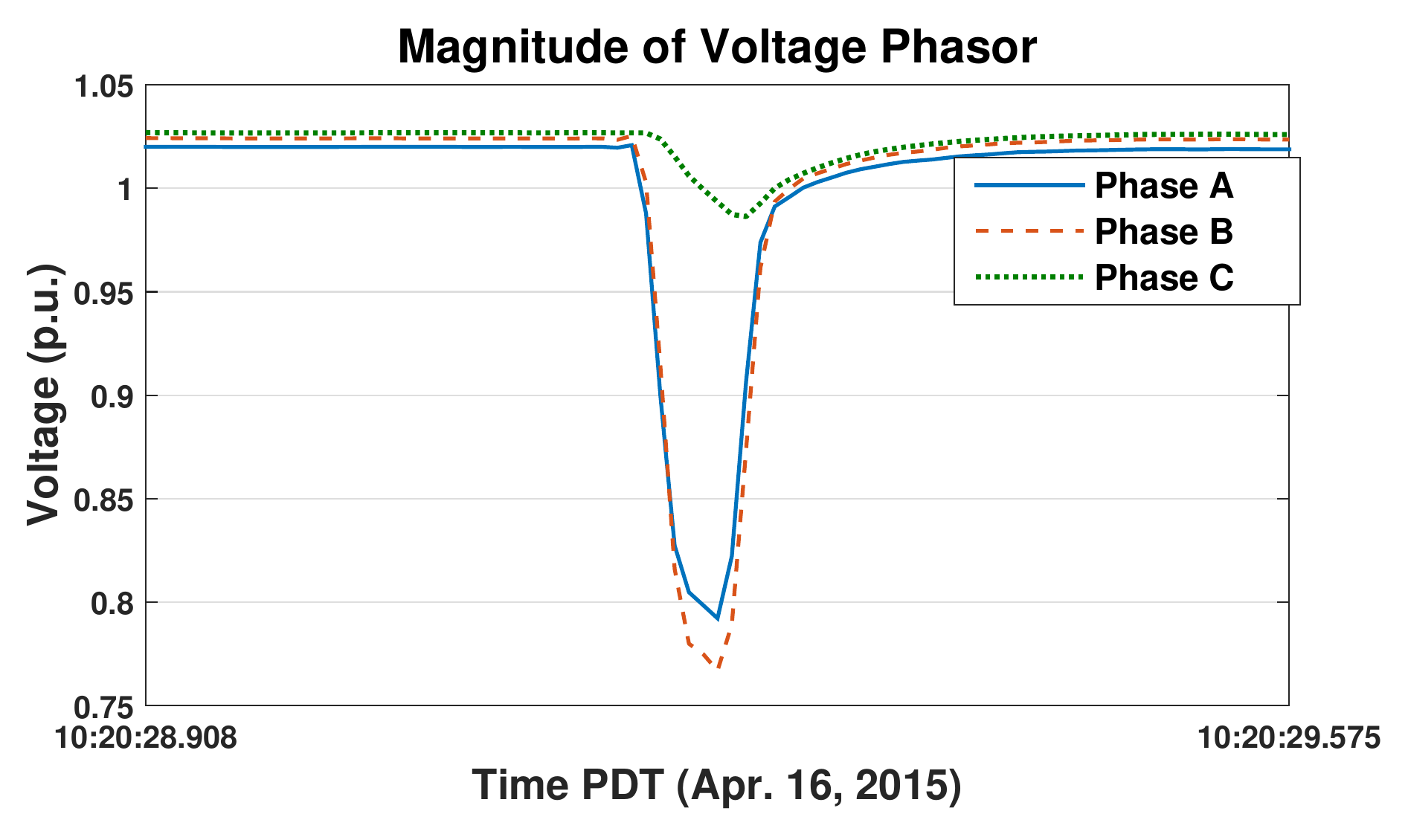}
\caption{Three Phase Voltage Phasor Magnitude Captured by a $\mup$ During the Voltage Sag}
\label{fig:vol_sag}
\end{figure} 

The $\mup$s used in this use case are installed in pairs at opposite ends of a line. We first show, for a specific line the change detection using the optimum cost function defined in \eqref{eq:cost_two_upmu} with two $\mup$s data, and then demonstrate that it would be possible to detect the changes with a single $\mup$ using \eqref{eq:cost_one_upmu}. The size of the null space using two and one $\mup$ during the voltage sag event is illustrated in Fig.~\ref{fig:cost_one_two_real}(a) and Fig.~\ref{fig:cost_one_two_real}(b) with $M=32$ samples, respectively \footnote{The voltage and current phasors are first converted to per-unit system assuming $S_b=1$ MVA.}. It can be observed that our metrics are able to effectively detect the anomaly as changes in the null-space, even when there is only one $\mup$ installed. The red markers are pointing to the detected anomaly start time, $\hat{k}_{c}$. Although the behavior in double and single $\mup$ metric is very similar, the change in the double-$\mup$ case is more pronounceable (noticing the scale on the vertical axis), obviously because it is augmenting the effects of the voltage sag on two $\mup$s rather than one.   
%In addition, we can observe that the normalized optimum cost functions for two $\mup$s and single $\mup$ are sufficiently similar to each other, therefore justifying the approximation we made in formulating the single $\mup$ detection rule. 
\begin{figure}[h]
\includegraphics[width=0.5\textwidth]{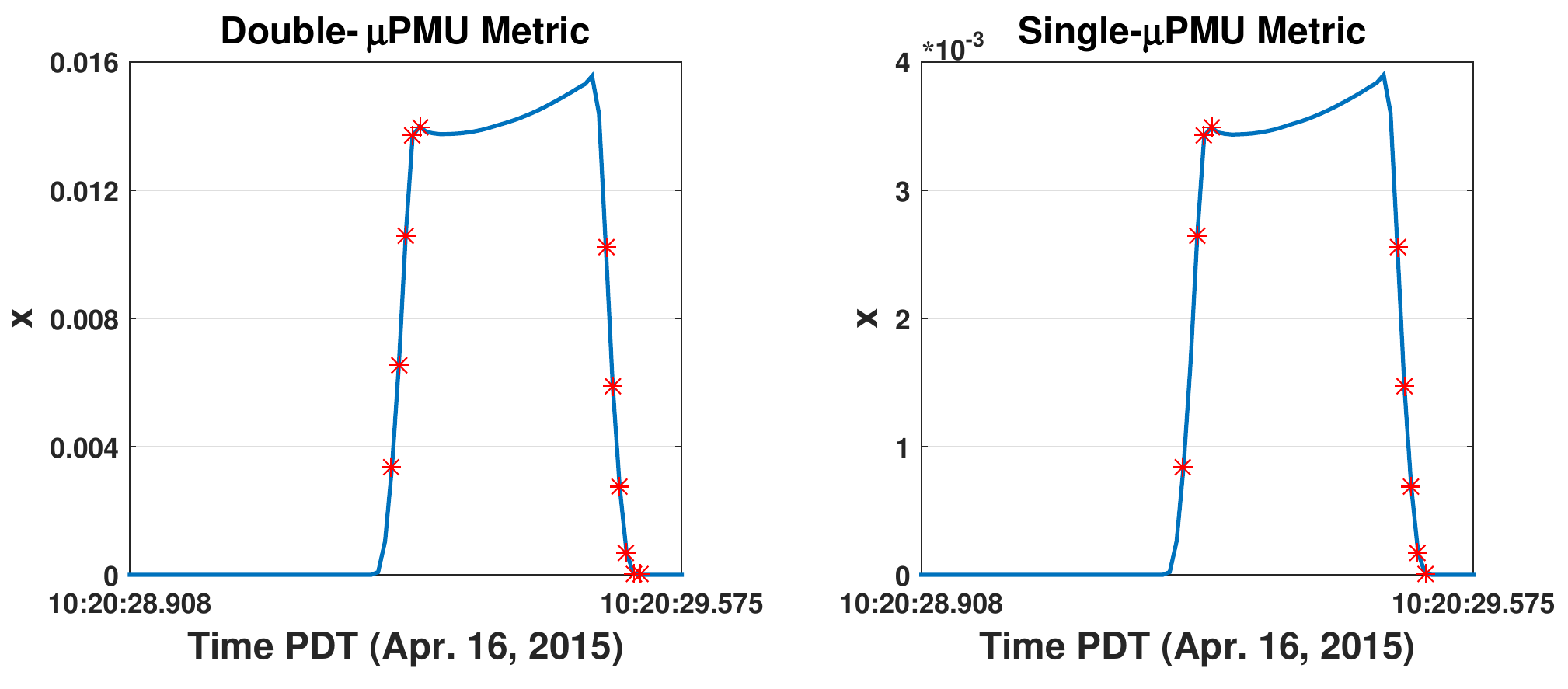}
\caption{Double and Single $\mup$ Metric Using Real Data for a Sag Event}
\label{fig:cost_one_two_real}
\end{figure}
It should be noted that since the value of the optimum cost function at time $k$ depends on the last $M$ phasor samples, there is a delay in the appearing and disappearing of the event in the cost functions.
\subsection{Simulated Data}
The IEEE 34-bus system is simulated in this section to test the rules on simulated data. The single-line diagram is shown in Fig.~\ref{fig:ieee34}, where the bus numbers are restarted from 1, compared to original feeder, for simplicity. The test feeder data can be found at \cite{ieee34}.
\begin{figure}[h]
\includegraphics[width=0.5\textwidth]{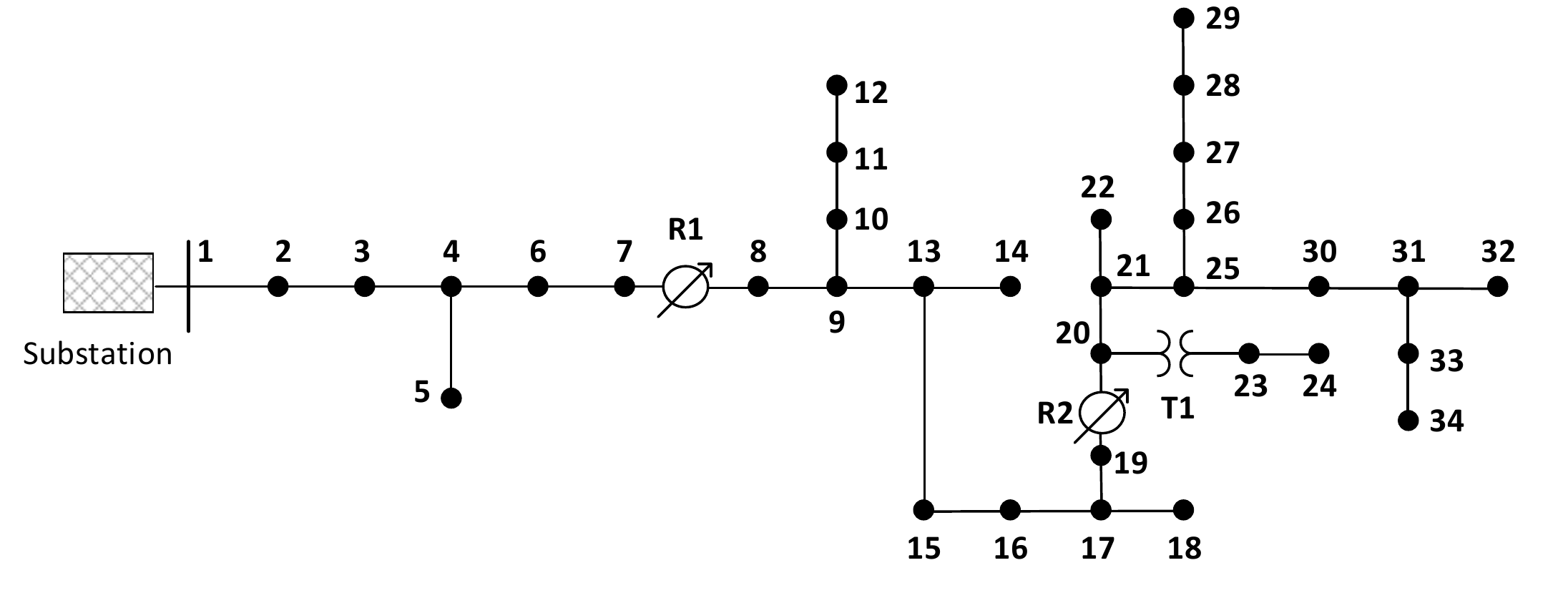}
\caption{IEEE 34-bus Test Feeder Single-Line Diagram}
\label{fig:ieee34}
\end{figure}
The simulation is performed in the time-series simulation environment of DIgSILENT \cite{manual2009version} and the time-domain waveforms are forwarded to our simulated $\mup$ model to obtain the phasor representation. We used the two-cycle P class filter in C.37.118 standard \cite{c37} in our simulated $\mup$ since we did not have access to the proprietary filters implemented within the $\mup$. This allows us to more closely mirror what a $\mup$ actually outputs in comparison to simulations that use FFT to extract the phasors. We consider a SLG fault scenario to assess how our rules perform. We assume that three $\mup$s, i.e., $K=3$, are available and are placed at bus 9, 19, and 31. The $\mup$ placement is done aiming to achieve the maximum change in \eqref{xopt_underdet} during a transient. The detailed formulation of the placement problem will be given in our future work. 

A temporary SLG fault at $50\%$ on phase A of line $(16,17)$ occurs at $t=0.5$ sec and is cleared at $t=0.52$ sec, before the recloser opens. Fig.~\ref{fig:cost_one_scen1} shows the single $\mup$ metric derived in \eqref{eq:cost_one_upmu} using the current on lines $(9,13)$, $(19,20)$ and $(31,32)$ with $M=6$ \footnote{The voltage and current phasors are first converted to per-unit system assuming $S_b=1$ MVA.}. The detected start time of the changes are marked for these three $\mup$s during the event using the same detection threshold and ``a priori'' knowledge about the magnitude of the change. As it can be observed, the number of changes found in each line's metric is correlated with how severely the event affects that particular $\mup$ measurement. The metric can be formed for all the incident lines to the buses with $\mup$s, however, we just show the cost relating to the aforementioned lines. 
\begin{figure}[h]
\includegraphics[width=0.5\textwidth]{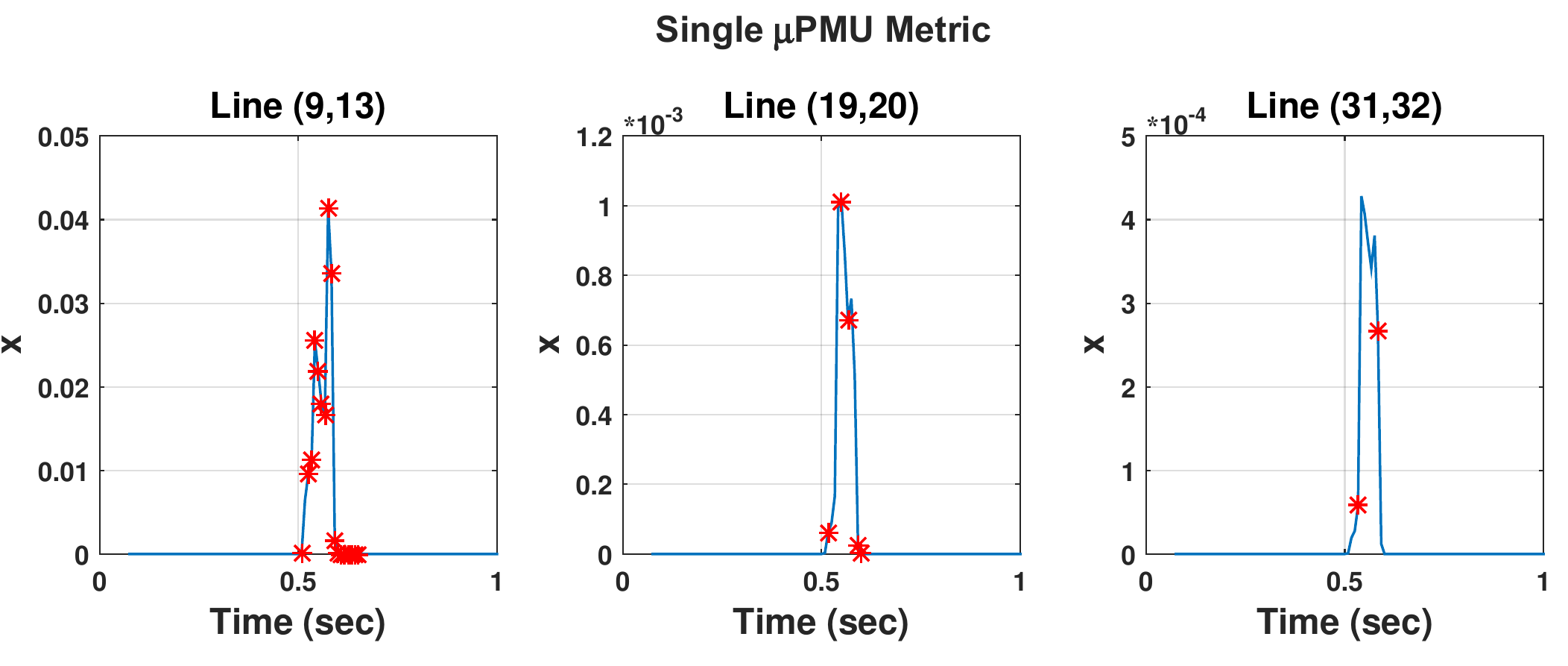}
\caption{Single $\mup$ Metric Using Simulated Data for SLG Fault}
\label{fig:cost_one_scen1}
\end{figure}
Clearly, in each case $x[k]$, calculated for each line, experiences a sudden change due to the fault event being sufficiently severe to shift all of the lines into a transient state. It is important to note that since phasor calculation is based on the past two-cycles of time-domain data, and since the cost function is affected by the $M-1$ previous phasor samples, as also the case for the real data, the anomaly appears and disappears with a delay in the cost function. To reduce the delay, the phasor can be estimated with less samples, e.g., one cycle instead of two cycles, but the associated accuracy decreases. In addition, $M$ can be set to the minimum possible value, which was the case here. The metric across multiple $\mup$s in \eqref{xopt_underdet} is also tested for this scenario and the corresponding metric, as well as the start time of the detected changes, are shown in Fig.~\ref{fig:rev_cost_cent_SLG}. As we expected, the fault manifests itself as a sudden change across the $\mup$s; a signature of an anomaly.
\begin{figure}[h]
\includegraphics[width=0.5\textwidth]{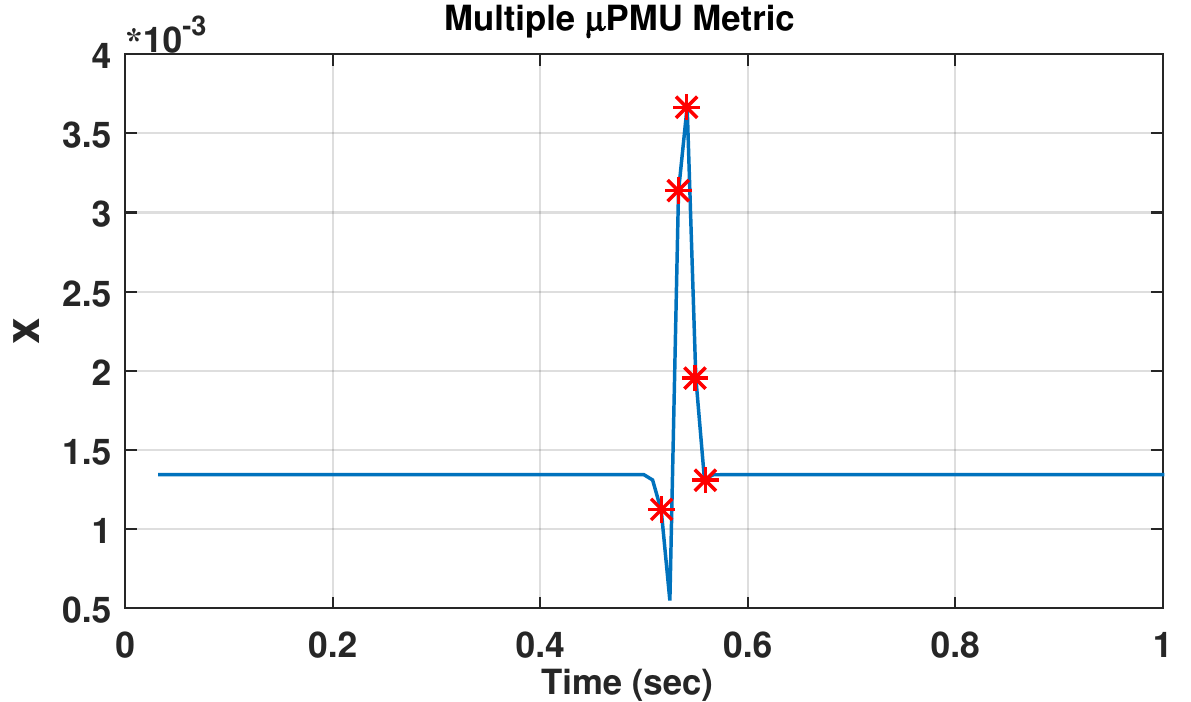}
\caption{Detected Changes in the Metric Across Multiple $\mup$s for Single-Line to Ground Fault}
\label{fig:rev_cost_cent_SLG}
\end{figure}
The delay here is solely due to the two-cycle calculation of the phasor, therefore less in comparison to the single $\mup$ metric.

\textit{Metrics Sensitivity to $\mup$ Data Manipulation:} An advantage of our metric across multiple $\mup$s is that it is robust to some degree to the $\mup$ data manipulation. In fact, as long as some of $\mup$ data are not compromised, the metric reveals the presence of a transient. What matters here is whether the change is sufficiently severe to trigger the change detector. If the detector is set to be too sensitive, the ``false alarms'' increase, so there is an associated trade-off that should be considered when the detector is designed. However, this study is beyond the scope of this paper. 

To test the performance of our metric in this situation, we consider again the single-line to ground fault previously explained for three cases where the attacker manipulates the data samples of the $\mup$ at bus 9 for the first case, the $\mup$ at bus 19 in the second case, and finally $\mup$s at buses 9 and bus 19 together for the third case. The attacker manipulates the data during the fault by pointing to the last sample before the fault starts. Fixing the detector parameters for all the three cases, Fig.~\ref{fig:rev_SLG_data_manipulate} shows the detected changes in the metric across multiple $\mup$s.
\begin{figure}[h]
\includegraphics[width=0.5\textwidth]{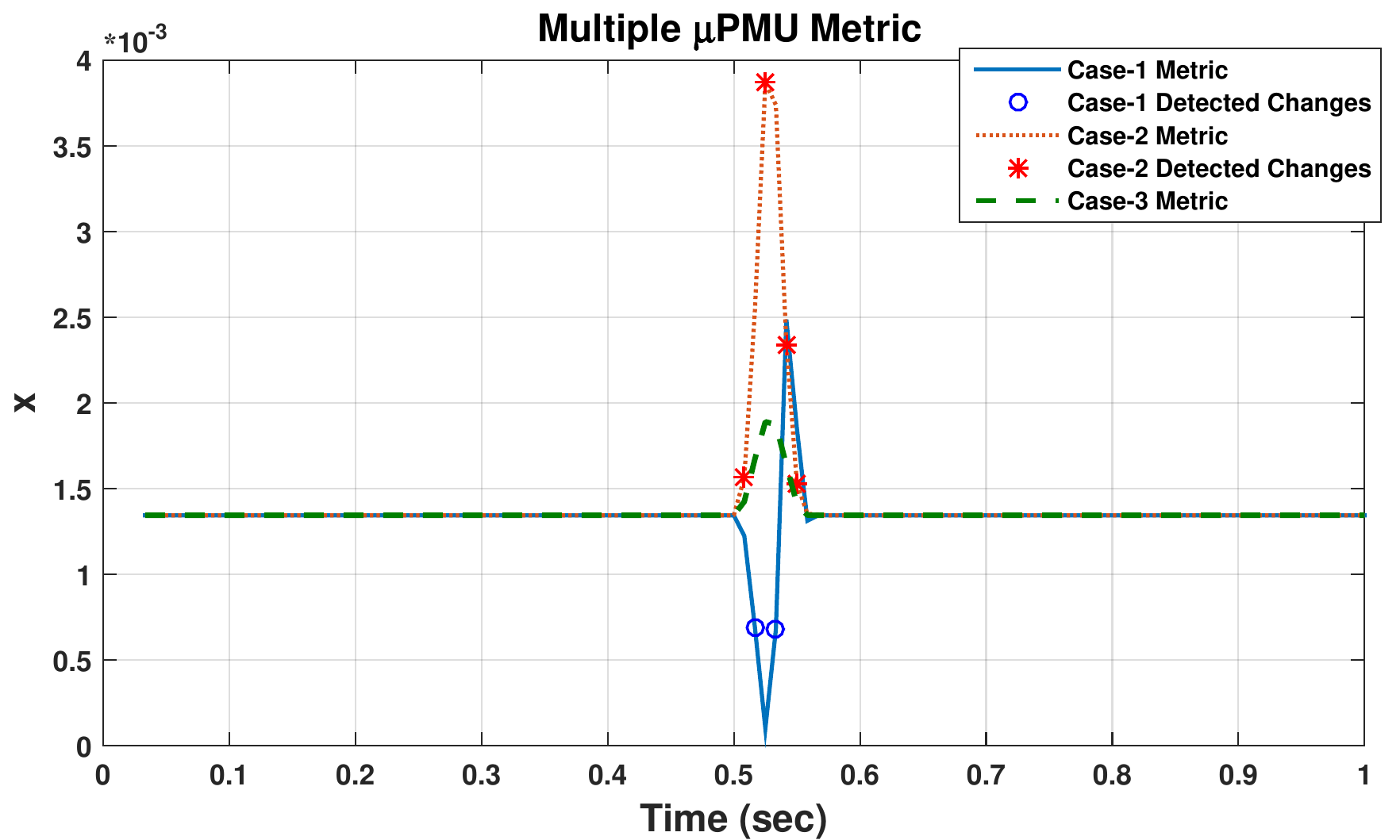}
\caption{Detected Changes in the Metric Across Multiple $\mup$s for Single-Line to Ground Fault under Manipulated $\mup$ Data}
\label{fig:rev_SLG_data_manipulate}
\end{figure}
It can be observed that the metric shows sudden changes in all the three cases. For the first and the second case, a few changes can still be detected by our CUSUSM detector. However, when both $\mup$ 9 and 19 are manipulated (case 3), the detector fails to spot the transient.

It is also clear that our single $\mup$ metric can only flag an anomaly if the corresponding $\mup$ data is not compromised. However, it should be mentioned that injecting false data at the device level is not an easy task for the attacker since the $\mup$ devices are designed to be read-only. In fact, man-in-the-middle attacks are more likely, which in this case, our single $\mup$ rule that is checked next to each device can detect the anomaly, even when the multiple $\mup$ rule is compromised. 

\textit{Metrics Sensitivity to Grid Connectivity Manipulation:}
The attacker can falsify the grid connectivity data to mislead our intrusion detector. Our single $\mup$ rule is set up to be agnostic about the grid interconnection, and therefore is not affected by such an attack. However, the grid connectivity becomes important for our rule using all the $\mup$ data.

\textit{Case--1:} We assume that grid connectivity data is manipulated, indicating that line $(25, 26)$ is out of service, while it is in service. The reason for choosing line $(25, 26)$ in our scenario is that no $\mup$ is installed on the lateral from bus 25 to 29 so we cannot confirm the line outage by just looking at the magnitude of the current phasor of that $\mup$. However, if a line outage is the ground truth, we expect to see a sudden change in our metrics because of the transient induced by the line switching before the step variation due to the change of topology. This expected transient is not observable in our metric, as can be seen in Fig.~\ref{fig:rev_Y_manip_c1}, and therefore this indication of a line outage can be flagged as an intrusion. 
\begin{figure}[h]
\includegraphics[width=8.0cm, height=4.5 cm]{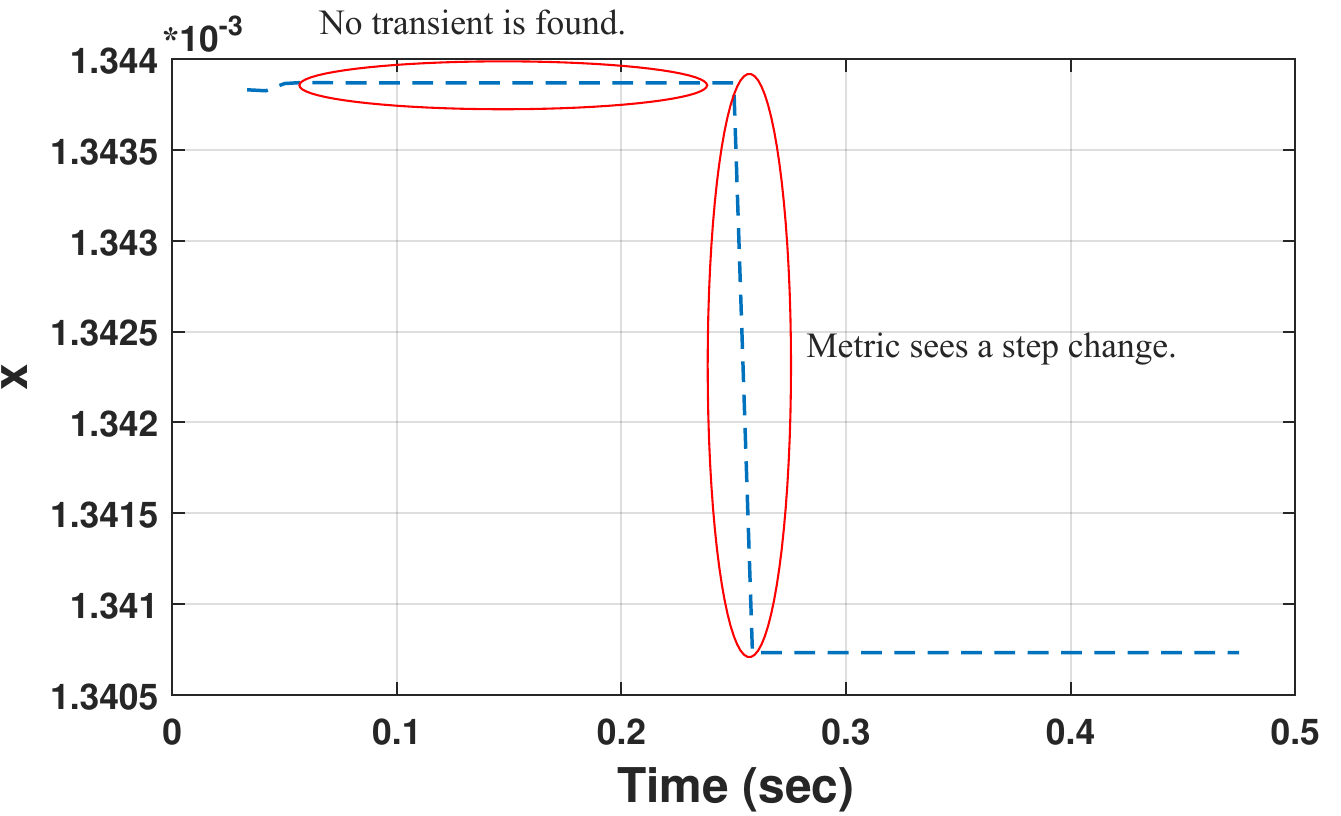}
\caption{Multiple $\mup$ Metric under Grid Topology Data Attack-Case 1}
\label{fig:rev_Y_manip_c1}
\end{figure} 

\textit{Case--2:} In this case, we assume that a three-phase fault occurs on line $(25, 26)$ at $t=0.4s$, resulting in the outage of the line due to the fuse operation at $t=0.46s$. Therefore the lateral $25-29$ is deenergized. The attacker falsifies the data, indicating that the line $(25, 26)$ is still in service after the transient finishes. Fig.~\ref{fig:rev_Y_manip_c2} shows the correct and the compromised metric.
An attack is undetectable by our multiple $\mup$ metric if the attacker is able to compromise some data, while not changing the value of the metric significantly before and after the fault. In our case, as we can observe in Fig.~\ref{fig:rev_Y_manip_c2}, the value of the compromised metric (post-fault value) is close to the ground truth (pre-fault value), and therefore this rule is not able to detect the associated data manipulation in this case.     
\begin{figure}[h]
\includegraphics[width=0.5\textwidth]{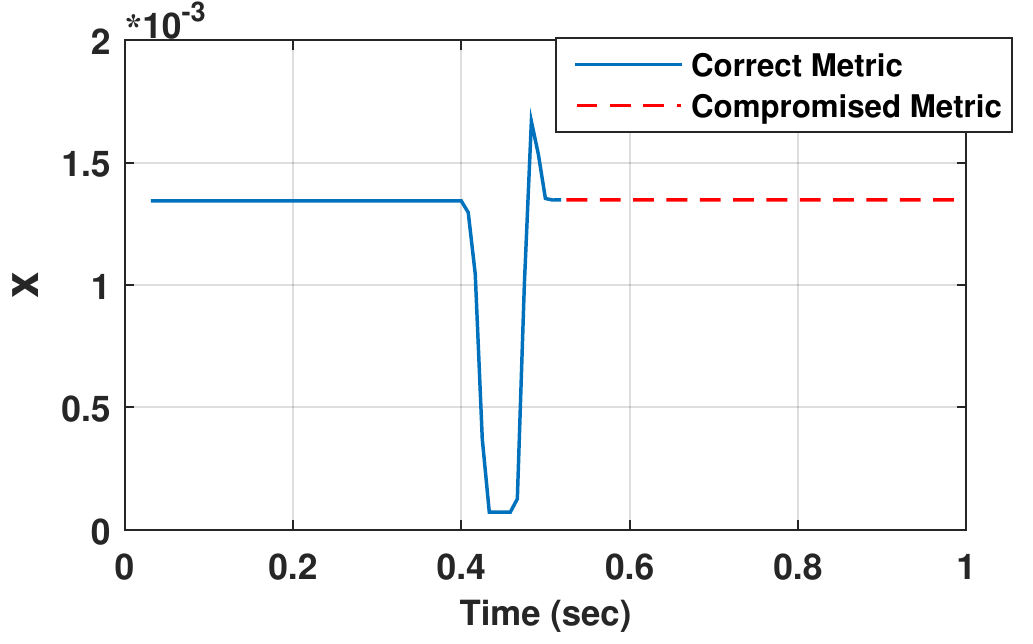}
\caption{Multiple $\mup$ Metric under Grid Topology Data Attack-Case 2}
\label{fig:rev_Y_manip_c2}
\end{figure}                 

%%%%%%%%%%%%%%%%%%%%%%%%%%%%%%%%%%%%%%%%%%%%%%%%
\section{Conclusion}
\label{sec:conc}
In this paper, we have formulated two types of metrics for anomaly detection using $\mup$ measurements for the distribution grid. The first metric requires phasor data of a single device and is agnostic about the grid model. The second metric correlates data across multiple $\mup$s to flag events. Anomalies are detected applying the sequential CUSUM algorithm in search for a sudden change in their mean value. To be more realistic, the rules have been formulated recognizing that the grid is in the quasi-steady state during the normal operation rather than operating at nominal frequency. In addition, the distribution grid has been modeled allowing for  unbalanced load, absent of the assumption of transposed lines and considering the fact that there may be two-phases or single phase laterals in the grid. The proposed rules have been verified using the real and synthetic data.

For future work, we will extend our rules to robustify the intrusion detector. For example, case-2 in the grid connectivity data manipulation can be identified as anomaly by checking the post and pre-fault power flow measured by the $\mup$ upstream of the event location. We will also present the optimal placement of the $\mup$ sensors with respect to our multiple $\mup$ metric to achieve the maximum change in the metric when an anomaly happens, as well as methods to localize the events that caused the transient. We will also propose a formal architecture for the implementation of our anomaly detectors in the context of cyber-physical security, where the analysis results of the $\mup$ data will be tied with the monitored DSCADA traffic for intrusion detection.       
%%%%%%%%%%%%%%%%%%%%%%%%%%%%%%%%%%%%%%%%%%%%%%%%
\section*{Acknowledgments}
This research was supported in part by the Director, Office
of Electricity Delivery and Energy Reliability, Cybersecurity
for Energy Delivery Systems program, of the U.S. Department
of Energy, under contracts DE-AC02-05CH11231 and DE-OE0000780. Any opinions in this material are those of the
authors and do not necessarily reflect those of the sponsors.
\bibliography{bib}
\vspace{-0.1cm}
\bibliographystyle{IEEEtran}
\end{document}